\documentclass[journal]{IEEEtran}

\usepackage{amsfonts,amscd,mathrsfs,amsmath,amsthm,amssymb}
\usepackage{mathtools}
\usepackage{xparse}
\usepackage{graphicx}
\usepackage{float}
\usepackage{pifont}
\usepackage[dvipsnames]{xcolor}
\usepackage{colortbl}
\usepackage{multirow}
\usepackage{enumerate}   
\usepackage{comment}
\usepackage{booktabs}
\usepackage{cancel}
\usepackage{listings}
\usepackage{xcolor}
\usepackage{Files/lieart}
\usepackage{Files/macros}

\usepackage{youngtab}

\PassOptionsToPackage{hyphens}{url}\usepackage{hyperref}

\hypersetup{
    colorlinks=true,
    linkcolor=magenta!80!black,  
    citecolor=magenta!80!black,
}
\usepackage{url} 

\usepackage[capitalise]{cleveref}
\usepackage{physics}
\usepackage{bm}
\usepackage{bbm}
\usepackage{caption}

\renewcommand{\SU}{\mathrm{SU}}

\newcommand{\Sym}{\mathrm{Sym}}
\renewcommand{\L}{\mathcal{L}} % linear operators

\renewcommand{\CC}{\mathbb{C}}
\renewcommand{\F}{\bmsf{F}}

\renewcommand{\Ad}{\bmsf{Ad}}
\renewcommand{\Sym}{\mathrm{Sym}}

\newcommand{\Hom}{\text{Hom}}

\newcommand{\lsuper}{\langle\!\langle}
\newcommand{\rsuper}{\rangle\!\rangle}

\renewcommand{\L}{\mathscr{L}} % linear operators
\renewcommand{\E}{W} 
\renewcommand{\B}{\mathcal{B}} % Error Basis
\newcommand{\proj}{\mathscr{P}} % superoperator proj
\newcommand{\twirl}{\mathscr{T}} % superoperator twirl

\renewcommand{\ip}[1]{\expval{#1}} %inner product on L(H)
\begin{document}

\title{Orthogonal Polynomials and the MacWilliams Transform for Permutation-Invariant Qudit Codes}

\author{ Ian~Teixeira

\thanks{Ian Teixeira is affiliated with the Department of Mathematics, University of California, San Diego, CA 92093}% <-this % stops a space
}

\maketitle

\begin{abstract}
We derive an explicit formula for the intrinsic MacWilliams transform for
permutation-invariant qudit codes.  Such codes naturally live in symmetric
power representations, where the relevant error sectors are determined by
the irreducible decomposition of the conjugation action on the associated
operator space.  Using the multiplicity-free structure of this decomposition
and the corresponding intertwiner algebra, we identify the intrinsic
MacWilliams matrix with a finite Racah transform.  The entries are given by
a terminating hypergeometric series, and the rows of the matrix are Racah
orthogonal polynomials with parameters determined explicitly by the block
length and local dimension.  Computing the spectrum of
the degree-one twirl reveals that this spectrum lies on an affine quadratic
lattice. Then we derive a tridiagonal multiplication rule from the
representation theory of the adjoint sector.  As consequences, we obtain
closed-form orthogonality, detailed-balance, and involutivity
identities for the transform.  The resulting formula supplies an explicit
MacWilliams matrix for computing linear programming bounds on permutation-invariant
qudit codes.
\end{abstract}
\tableofcontents

\section{Introduction}

MacWilliams identities are a central structural tool in coding theory:
they relate the weight distribution of a code to that of its dual and,
together with positivity constraints, lead to Delsarte-type linear
programming bounds through the theory of association schemes
\cite{MacWilliams1963,MacWilliamsSloane1977,Delsarte1973,BannaiIto1984}.
Quantum analogues of these identities were developed using quantum weight
enumerators and have played an important role in linear programming bounds
for quantum codes
\cite{ShorWE,rainsWE,GF4codes,Shadows,AshikhminLitsyn1999}.

Permutation-invariant quantum codes are a natural class of nonadditive
quantum codes whose code space is contained in the permutation-invariant sector of \(n\) qudits. They were introduced in the study of exchange
errors and symmetric quantum error correction
\cite{RuskaiPRL,2004permutation}, and have since been developed in several
directions, including higher-dimensional logical spaces, qudit codes,
amplitude-damping codes, and deletion-error models
\cite{PI2,PImore,PIqudit1,OuyangChao,PIdelete,ShibayamaHagiwara,PI3}.
For local dimension \(q\), the permutation-invariant subspace is
\(V_n=\Sym^n(\mathbb C^q)\).  In the intrinsic framework of
\cite{usIntrinsicCodes,usIntrinsicMacWilliams}, the relevant weight
enumerators are indexed not by ordinary Hamming weights, but by the
irreducible sectors appearing in the conjugation representation of
\(\SU(q)\) on \(\L(V_n)\).  The purpose of this paper is to compute the
corresponding intrinsic MacWilliams transform explicitly, thereby providing
the transform matrix needed for linear programming bounds on
permutation-invariant qudit codes.

The representation \(V_n=\Sym^n(\mathbb C^q)\) is the irreducible
\(\SU(q)\)-module of highest weight \((n,0,\dots,0)\)
\cite{fultonharris}.  Under the conjugation action of \(G=\SU(q)\), its
operator space has the multiplicity-free decomposition
\[
\L(V_n)\cong \bigoplus_{a=0}^n \E_a,
\qquad
\E_a\cong(a,0,\dots,0,a).
\]
Thus the irreducible sectors \(\E_a\) play, for permutation-invariant
qudit codes, the role played by Hamming-weight sectors in the usual
multi-qudit setting.  In particular, the intrinsic weight enumerators of
\cite{usIntrinsicMacWilliams} are indexed by these sectors.

Let \(\proj_a\) denote the orthogonal projector onto \(\E_a\), and let \(\twirl_b\) denote the twirling operator obtained by summing over an orthonormal basis of \(\E_b\). The projectors
\(\{\proj_a\}\) and twirls \(\{\twirl_b\}\) form two natural bases of the
commutative intertwiner algebra
\[
\mathcal A=\Hom_G(\L(V_n),\L(V_n)).
\]
The intrinsic MacWilliams matrix \(M=(M_{ba})\) is the change-of-basis
matrix between them:
\[
\twirl_b=\sum_{a=0}^n M_{ba}\proj_a.
\]
Equivalently, \(M_{ba}\) is the scalar by which \(\twirl_b\) acts on the
sector \(\E_a\). Thus the problem of computing the intrinsic MacWilliams transform becomes
the spectral problem of determining the joint eigenvalues of the twirling
operators on the irreducible sectors \(\E_a\).

Our main result is that this intrinsic MacWilliams matrix is a finite Racah
transform, in the sense of the Racah orthogonal polynomial system
\cite{KoekoekLeskySwarttouw2010,Wilson1980}.  More precisely, with
\(d_b=\dim\E_b\) and \(N=\dim V_n=\binom{n+q-1}{n}\), we prove that
\[
M_{ba}
=
\frac{d_b}{N}
\;
{}_4F_3\!\left(
\begin{matrix}
-b,\ b+q-1,\ -a,\ a+q-1\\
q-1,\ -n,\ n+q
\end{matrix}
;1
\right).
\]
Equivalently, after an affine change of variables
\(x_a=A-Ba(a+q-1)\), the row polynomials \(p_b(x_a)=M_{ba}\) are Racah
polynomials with parameters
\[
\alpha=q-2,\qquad
\beta=0,\qquad
\gamma=n+q-1,\qquad
\delta=-n-1.
\]

The proof has four main ingredients.  First, the multiplicity-free
decomposition of \(\L(V_n)\) gives the projector and twirl bases of the
intertwiner algebra \(\mathcal A\).  Second, we compute the spectrum of the
degree-one twirl \(\twirl_1\) by expressing it as an affine function of the
conjugation Casimir.  The resulting spectral grid is an affine image of the
quadratic lattice \(a(a+q-1)\), and the distinctness of these eigenvalues
shows that \(\twirl_1\) generates \(\mathcal A\).  Third, a Pieri-rule
support lemma shows that multiplication by \(\twirl_1\) is tridiagonal in
the twirl basis.  Hence the row functions \(p_b\) satisfy a three-term
recurrence and form a finite orthogonal polynomial system
\cite{Chihara1978,Szego1975}.  Finally, comparison with the Racah lattice
and Racah measure identifies this system with Racah polynomials and gives
the stated \({}_4F_3\) formula.

As consequences, the transform satisfies the weighted orthogonality relation
\(MDM^T=D\), detailed balance \(d_aM_{ba}=d_bM_{ab}\), and the involutivity
relation \(M^2=I\), paralleling the structural identities of classical
MacWilliams transforms
\cite{MacWilliams1963,MacWilliamsSloane1977,Delsarte1973,BannaiIto1984}.
Thus the intrinsic MacWilliams transform for permutation-invariant qudit
codes is explicitly computable, self-inverse, and governed by the classical
Racah orthogonal polynomial system.  In particular, it gives a closed-form
transform matrix for Delsarte-type linear programming bounds in the
permutation-invariant qudit setting.

\section{Representation-Theoretic Setup}
\label{sec:setup}

Let
\[
  V_n := \Sym^n(\mathbb C^q)
\]
be the \(n\)th symmetric power of the defining representation of
\(G:=\SU(q)\).  This is the irreducible \(G\)-representation of highest
weight \((n,0,\dots,0)\), and
\begin{equation}
  \label{eq:dimVn}
  \dim V_n=\binom{n+q-1}{n}.
\end{equation}
Write \(\L(V_n)\) for the space of linear operators on \(V_n\), equipped
with the Hilbert--Schmidt inner product
\begin{equation}
  \label{eq:hs-ip}
  \ip{X_1,X_2} := \Tr(X_1^\dagger X_2),
  \qquad X_1,X_2\in \L(V_n).
\end{equation}
The group \(G\) acts on \(\L(V_n)\) by conjugation:
\begin{equation}
  \label{eq:conj-action}
  g\cdot X := \rho_n(g)X\rho_n(g)^\dagger,
  \qquad g\in G,\; X\in \L(V_n),
\end{equation}
where \(\rho_n\) denotes the representation of \(G\) on \(V_n\).

\subsection{Multiplicity-free decomposition of \texorpdfstring{\(\L(V_n)\)}{L(Vn)}}

The conjugation representation on
\(\L(V_n)\cong V_n\otimes V_n^*\) is multiplicity-free.  Concretely,
\begin{equation}
  \label{eq:Lv-decomp}
  \L(V_n)\cong \bigoplus_{a=0}^n \E_a,
  \qquad
  \E_a\cong (a,0,\dots,0,a),
\end{equation}
where \((a,0,\dots,0,a)\) denotes the irreducible \(\SU(q)\)-representation
with that Dynkin label.  We write
\[
  d_a := \dim \E_a.
\]
For each \(a\), let
\[
  \proj_a:\L(V_n)\to \L(V_n)
\]
denote the orthogonal projector onto \(\E_a\) with respect to
\eqref{eq:hs-ip}.  These projectors satisfy
\begin{equation}
  \label{eq:projector-relations}
  \proj_a\proj_b=\delta_{ab}\proj_a,
  \qquad
  \sum_{a=0}^n \proj_a=I_{\L(V_n)}.
\end{equation}

\subsection{The intertwiner algebra}

Let
\begin{equation}
  \label{eq:intertwiner-algebra}
  \mathcal A := \Hom_G\bigl(\L(V_n),\L(V_n)\bigr)
\end{equation}
denote the algebra of \(G\)-equivariant endomorphisms of \(\L(V_n)\).
We equip \(\L(\L(V_n))\) with the Hilbert--Schmidt
inner product on superoperators:
\begin{equation}
  \label{eq:A-ip-def}
  \lsuper S,T \rsuper
  := \Tr_{\L(V_n)}(S^\dagger T)
  = \sum_i \ip{E_i, S^\dagger T(E_i)},
\end{equation}
where \(\{E_i\}\) is any orthonormal basis of \((\L(V_n),\ip{\cdot,\cdot})\),
and \(S^\dagger\) denotes the adjoint with respect to \eqref{eq:hs-ip}.

Since the decomposition \eqref{eq:Lv-decomp} is multiplicity-free,
Schur's lemma implies that every \(T\in\mathcal A\) acts as a scalar on each
irreducible sector \(\E_a\).  Hence
\[
  \dim \mathcal A=n+1,
\]
and \(\mathcal A\) is commutative.  Moreover, the normalized projectors
\[
  \{d_a^{-1/2}\proj_a\}_{a=0}^n
\]
form an orthonormal basis of
\((\mathcal A,\lsuper \cdot, \cdot \rsuper)\).

\subsection{Twirling operators}

For each \(0\le b\le n\), the \emph{twirling operator} associated with
\(\E_b\) is
\begin{equation}
  \label{eq:twirl-def}
  \twirl_b(X) := \sum_{E\in \B_b} E^\dagger X E,
  \qquad X\in \L(V_n),
\end{equation}
where \(\B_b\) is any orthonormal basis of \(\E_b\) with respect to
\eqref{eq:hs-ip}.  By \cite[Section~III]{usIntrinsicMacWilliams},
the definition \eqref{eq:twirl-def} is independent of the choice of
\(\B_b\), each \(\twirl_b\) is \(G\)-equivariant, and the normalized twirls
\[
  \{d_b^{-1/2}\twirl_b\}_{b=0}^n
\]
form a second orthonormal basis of 
\((\mathcal A,\lsuper \cdot, \cdot \rsuper)\)
.

\subsection{The intrinsic MacWilliams matrix}

Since \(\{\proj_a\}_{a=0}^n\) and \(\{\twirl_b\}_{b=0}^n\) are two bases of
\(\mathcal A\), there is a unique \((n+1)\times(n+1)\) matrix
\(M=(M_{ba})_{0\le b,a\le n}\) such that
\begin{equation}
  \label{eq:twirl-expansion}
  \twirl_b=\sum_{a=0}^n M_{ba}\proj_a.
\end{equation}
Restricting \eqref{eq:twirl-expansion} to the sector \(\E_a\) gives
\begin{equation}
  \label{eq:twirl-on-sector}
  \twirl_b\big|_{\E_a}=M_{ba}I_{\E_a}.
\end{equation}
Thus \(M_{ba}\) is the scalar by which \(\twirl_b\) acts on \(\E_a\).
Equivalently, the \(a\)th column of \(M\) records the joint spectrum of the
commuting family \(\{\twirl_b\}_{b=0}^n\) on \(\E_a\).  We call \(M\) the
\emph{intrinsic MacWilliams matrix}.

Let
\begin{equation}
  \label{eq:D-def}
  D:=\operatorname{diag}(d_0,d_1,\dots,d_n).
\end{equation}
Because both
\[
  \{d_a^{-1/2}\proj_a\}_{a=0}^n
  \qquad\text{and}\qquad
  \{d_b^{-1/2}\twirl_b\}_{b=0}^n
\]
are orthonormal bases of the same Hilbert space
\((\mathcal A,\lsuper \cdot, \cdot \rsuper)\), the normalized
change-of-basis matrix
\[
  D^{-1/2}MD^{1/2}
\]
is unitary.  Equivalently, \(M\) satisfies the Plancherel identity
\begin{equation}
  \label{eq:plancherel}
  M D \overline{M}^{\,T}=D,
\end{equation}
and hence the inversion formula
\begin{equation}
  \label{eq:M-inverse-setup}
  M^{-1}=D\overline{M}^{\,T}D^{-1}.
\end{equation}
The explicit formula derived in \Cref{sec:structure} shows that the entries
of \(M\) are real, so \eqref{eq:plancherel} and \eqref{eq:M-inverse-setup}
reduce to
\[
  MDM^T=D,
  \qquad
  M^{-1}=DM^TD^{-1}.
\]

\section{Eigenvalues of the Degree-One Twirl}
\label{sec:x_a}

The main result of this section is an explicit closed-form expression for
the eigenvalue of the degree-one twirling operator \(\twirl_1\) on each
irreducible sector \(\E_a \subset \L(V_n)\).  The proof compares two
quadratic Casimirs: the ordinary Casimir \(C\) acting on \(V_n\), and the
conjugation Casimir \(\mathscr C\) acting as a superoperator on \(\L(V_n)\).

\begin{lemma}[Eigenvalues of \(\twirl_1\)]
\label{lem:twirl1-eigenvalues}
For \(q \ge 2\) and \(0 \le a \le n\), the twirling operator
\(\twirl_1\) acts on the irreducible sector \(\E_a \subset \L(V_n)\) as the
scalar
\begin{equation}
  \label{eq:xa-final}
  x_a \;:=\; M_{1a}
  \;=\;
  \frac{q^2 - 1}{\binom{n+q-1}{n}}
  \left(
    1 -
    \frac{a(a+q-1)q}{n(q-1)(n+q)}
  \right).
\end{equation}
The spectral values of \(\twirl_1\) lie on an affine image of the quadratic
lattice \(\{a(a+q-1):0\le a\le n\}\), and \(a\mapsto x_a\) is strictly
decreasing.
\end{lemma}

\begin{proof}
We use the Hermitian-generator convention for \(\mathfrak{su}(q)\): thus
\(\mathfrak{su}(q)\) denotes the real vector space of traceless Hermitian
\(q\times q\) matrices, with Lie bracket \(i[A,B]\).  Fix an orthonormal
basis
\[
  \{T_\mu\}_{\mu=1}^{q^2-1}\subset \mathfrak{su}(q)
\]
with respect to the Hilbert--Schmidt form, so $ \Tr(T_\mu T_\nu)=\delta_{\mu\nu}$. 
Let \(J_\mu=d\rho_n(T_\mu)\in \L(V_n)\) denote the corresponding Hermitian
generators in the representation of $ \mathfrak{su}(q) $ on \(V_n=\Sym^n(\mathbb C^q)\).

\medskip\noindent\textbf{Step 1: two Casimirs.}
There are two quadratic Casimirs in play.  The ordinary Casimir acts on
\(V_n\).  Since \(V_n\) is irreducible, Schur's lemma gives
\begin{equation}
  \label{eq:ordinary-casimir}
  C:=\sum_{\mu=1}^{q^2-1}J_\mu^2=c_V I_{V_n}.
\end{equation}
The conjugation Casimir is the superoperator on \(\L(V_n)\) defined by
\begin{equation}
  \label{eq:conj-casimir}
  \mathscr C(X):=\sum_{\mu=1}^{q^2-1}[J_\mu,[J_\mu,X]].
\end{equation}
It is \(G\)-equivariant for the conjugation representation.  Hence, since
each \(\E_a\) is \(G\)-irreducible, Schur's lemma gives
\[
  \mathscr C\big|_{\E_a}=c_a I_{\E_a}
\]
for some scalar \(c_a\).

\medskip\noindent\textbf{Step 2: expressing \(\twirl_1\) in terms of \(\mathscr C\).}
The image \(d\rho_n(\mathfrak{su}(q))\) is the adjoint sector
\(\E_1\subset \L(V_n)\).  Since \(\E_1\) is irreducible and the
Hilbert--Schmidt form is \(G\)-invariant, Schur's lemma gives a scalar
\(\kappa_n>0\) such that
\begin{equation}
  \label{eq:kappa-def}
  \Tr_{V_n}(J_\mu J_\nu)=\kappa_n\delta_{\mu\nu}.
\end{equation}
Equivalently, \(\kappa_n\) is the Dynkin index of \(V_n\) in this
normalization.  Therefore
\(\{J_\mu/\sqrt{\kappa_n}\}_{\mu=1}^{q^2-1}\) is an orthonormal basis of
\(\E_1\), and the definition of \(\twirl_1\) gives
\begin{equation}
  \label{eq:twirl1-J}
  \twirl_1(X)
  =
  \frac{1}{\kappa_n}
  \sum_{\mu=1}^{q^2-1}J_\mu XJ_\mu.
\end{equation}
Expanding the double commutator in \eqref{eq:conj-casimir} and using
\eqref{eq:ordinary-casimir}, we obtain
\begin{align}
  \mathscr C(X)
  &=
  \Bigl(\textstyle\sum_\mu J_\mu^2\Bigr)X
  -2\sum_\mu J_\mu XJ_\mu
  +X\Bigl(\textstyle\sum_\mu J_\mu^2\Bigr) \nonumber\\
  &=
  2c_VX-2\sum_{\mu=1}^{q^2-1}J_\mu XJ_\mu.
  \label{eq:conj-casimir-expanded}
\end{align}
Comparing \eqref{eq:conj-casimir-expanded} with \eqref{eq:twirl1-J} gives
\begin{equation}
  \label{eq:twirl1-casimir}
  \twirl_1
  =
  \frac{c_V}{\kappa_n}I_{\L(V_n)}
  -
  \frac{1}{2\kappa_n}\mathscr C,
\end{equation}
where \(I_{\L(V_n)}\) denotes the identity superoperator on \(\L(V_n)\).

\medskip\noindent\textbf{Step 3: the Dynkin index.}
Taking the trace of \eqref{eq:ordinary-casimir} on \(V_n\), and using
\eqref{eq:kappa-def}, gives
\[
  (q^2-1)\kappa_n
  =
  \sum_{\mu=1}^{q^2-1}\Tr_{V_n}(J_\mu^2)
  =
  c_V\dim V_n.
\]
Thus
\begin{equation}
  \label{eq:kappa-value}
  \kappa_n
  =
  \frac{c_V\dim V_n}{q^2-1}
  =
  \frac{c_V\dim V_n}{\dim\mathfrak{su}(q)}.
\end{equation}
Substituting this into \eqref{eq:twirl1-casimir} yields
\begin{equation}
  \label{eq:twirl1-final}
  \twirl_1
  =
  \frac{q^2-1}{\dim V_n}
  \left(
    I_{\L(V_n)}
    -
    \frac{1}{2c_V}\mathscr C
  \right).
\end{equation}

\medskip\noindent\textbf{Step 4: eigenvalues.}
Restricting \eqref{eq:twirl1-final} to \(X\in \E_a\) and using
\(\mathscr C|_{\E_a}=c_aI_{\E_a}\) gives
\begin{equation}
  \label{eq:xa-casimir}
  x_a
  =
  \frac{q^2-1}{\dim V_n}
  \left(1-\frac{c_a}{2c_V}\right)
  =
  \frac{\dim\mathfrak{su}(q)}{\dim V_n}
  \left(1-\frac{c_a}{2c_V}\right).
\end{equation}
By Appendix~\ref{app:casimir}, in the normalization
\(\Tr_{\mathbb C^q}(T_\mu T_\nu)=\delta_{\mu\nu}\), the relevant Casimir
eigenvalues are
\begin{equation}
  \label{eq:cV-ca}
  c_V=\frac{n(q-1)(n+q)}{q},
  \qquad
  c_a=2a(a+q-1).
\end{equation}
Therefore
\[
  \frac{c_a}{2c_V}
  =
  \frac{q\,a(a+q-1)}{n(q-1)(n+q)}.
\]
Using \(\dim V_n=\binom{n+q-1}{n}\), this gives \eqref{eq:xa-final}.

Finally, \(a\mapsto a(a+q-1)\) is strictly increasing on
\(\{0,1,\dots,n\}\) for \(q\ge 2\), and its coefficient in
\eqref{eq:xa-final} is strictly negative.  Hence \(a\mapsto x_a\) is
strictly decreasing.
\end{proof}

Setting \(y_a:=a(a+q-1)\), the eigenvalues take the affine form
\[
  x_a=A-By_a,
\]
where
\[
  A=\frac{q^2-1}{\dim V_n},
  \qquad
  B=\frac{(q^2-1)q}{\dim V_n\,n(q-1)(n+q)}.
\]
This affine-quadratic-lattice structure reappears in
\Cref{sec:spectral-poly} as the spectral grid of the Racah polynomial
system.

The case \(a=0\) gives a useful check.  Since \(\E_0=\mathbb C I_{V_n}\),
we have \(\twirl_1(I_{V_n})=x_0I_{V_n}\).  Taking traces gives
\[
  x_0\dim V_n
  =
  \Tr(\twirl_1(I_{V_n}))
  =
  \sum_{E\in \B_1}\Tr(E^\dagger E)
  =
  \dim \E_1
  =
  q^2-1,
\]
so
\[
  x_0=\frac{q^2-1}{\dim V_n},
\]
in agreement with \eqref{eq:xa-final}.

For \(q=2\), \(\dim V_n=n+1\), and \eqref{eq:xa-final} reduces to
\[
  x_a
  =
  \frac{3}{n+1}
  \left(
    1-\frac{2a(a+1)}{n(n+2)}
  \right),
  \qquad a=0,1,\dots,n.
\]
This agrees with the \(b=1\) row of the \(\SU(2)\) intrinsic MacWilliams
matrix in \cite{usIntrinsicMacWilliams}.

\section{The Spectral Model}
\label{sec:spectral-poly}

The distinct eigenvalues of $\twirl_1$ established in \Cref{sec:x_a}
allow us to identify the intertwiner algebra $\mathcal{A}$ with a polynomial
algebra and to realize each twirling operator $\twirl_b$ as a polynomial in
$\twirl_1$ evaluated on a finite spectral grid.
The resulting evaluation matrix is precisely the MacWilliams matrix $M$,
and the Plancherel identity forces the row polynomials $\{p_b\}$ to be
orthogonal with respect to a natural discrete measure on the grid.
We establish here that these polynomials have degree at most $n$; the
sharper statement $\deg p_b = b$, which elevates $\{p_b\}$ to a genuine
orthogonal polynomial system, is a consequence of the tridiagonality result
proved in \Cref{sec:tridiagonal-ops}.

\subsection{Joint diagonalization and the spectral algebra isomorphism}
\label{sec:x_a-joint-diag}

By multiplicity-freeness and Schur's lemma, every $T\in\mathcal{A}$ acts as
a scalar on each irreducible sector $\E_a$:
\[
  T\big|_{\E_a} = \lambda_a(T)\,I_{\E_a}. \numberthis
\]
Together with $\sum_{a=0}^n \proj_a = I_{\L(V_n)}$, this gives the
spectral expansion
\begin{equation}
  \label{eq:T-diagonal-spectral}
  T = \sum_{a=0}^n \lambda_a(T)\,\proj_a, \qquad T \in \mathcal{A}.
\end{equation}
\begin{lemma}
  \label{lem:lambda-iso}
  The map
  \begin{equation}
    \label{eq:lambda-map-spectral}
    \lambda : \mathcal{A} \longrightarrow \mathbb{C}^{n+1},
    \qquad T \longmapsto \bigl(\lambda_0(T),\dots,\lambda_n(T)\bigr),
  \end{equation}
  is an isomorphism of algebras, where $\mathbb{C}^{n+1}$ carries
  coordinatewise multiplication.
\end{lemma}
\begin{proof}
  Linearity and bijectivity are clear: $\lambda$ is $\mathbb{C}$-linear,
  injective by~\eqref{eq:T-diagonal-spectral} (if all $\lambda_a(T) = 0$
  then $T = 0$), and hence bijective since $\dim\mathcal{A} = n+1$.
  For multiplicativity, let $S, T \in \mathcal{A}$.
  Since $T|_{\E_a} = \lambda_a(T)\,I_{\E_a}$ and $S$ preserves $\E_a$,
  \[
    (S \circ T)\big|_{\E_a}
    = S\bigl(\lambda_a(T)\,I_{\E_a}\bigr)
    = \lambda_a(T)\,\lambda_a(S)\,I_{\E_a}, \numberthis
  \]
  so $\lambda_a(S \circ T) = \lambda_a(S)\,\lambda_a(T)$, as required.
\end{proof}

For the twirling operators, the MacWilliams coefficients are precisely
these joint eigenvalues:
\begin{equation}
  \label{eq:M-as-eigenvalues-spectral}
  M_{ba} = \lambda_a(\twirl_b).
\end{equation}
Thus the $b$th row of $M$ records the eigenvalues of $\twirl_b$ on the
sectors $\E_a$.

\subsection{The spectral grid of \texorpdfstring{$\twirl_1$}{twirl\_1}}
\label{sec:spectral-grid}

From \Cref{sec:x_a}, the degree-one twirl has spectral decomposition
\begin{equation}
  \label{eq:twirl1-spectral}
  \twirl_1 = \sum_{a=0}^n x_a\,\proj_a,
\end{equation}
where
\begin{equation}
  \label{eq:x-a-spectral}
  x_a
  =
  \frac{q^2-1}{\binom{n+q-1}{n}}
  \!\left(
    1-\frac{q\,a(a+q-1)}{n(q-1)(n+q)}
  \right).
\end{equation}
We call $\{x_0,\dots,x_n\}$ the \emph{spectral grid} of $\twirl_1$.

\begin{lemma}
  \label{lem:x-distinct}
  The scalars $x_0,x_1,\dots,x_n$ are pairwise distinct.
\end{lemma}

\begin{proof}
  Set $y_a := a(a+q-1)$.  Since for \(a \ge 0,\; q \ge 2\) 
  \[
    y_{a+1} - y_a = (a+1)(a+q) - a(a+q-1) = 2a+q > 0
    \numberthis
  \]
   the sequence $(y_a)$ is strictly increasing.  The coefficient of $y_a$
  in~\eqref{eq:x-a-spectral} is
  $-q(q^2-1)/\bigl[\binom{n+q-1}{n}n(q-1)(n+q)\bigr] < 0$,
  so $a \mapsto x_a$ is strictly decreasing, giving pairwise distinct values.
\end{proof}

Since $\twirl_1$ has $n+1$ distinct eigenvalues $x_0,\dots,x_n$ as a
superoperator on $\L(V_n)$, its minimal polynomial is
\begin{equation}
  \label{eq:min-poly-twirl1}
  m(x) = \prod_{a=0}^n (x - x_a).
\end{equation}
Consequently every polynomial in $\twirl_1$ has a unique representative of
degree at most $n$, and $\mathbb{C}[\twirl_1] \cong \mathbb{C}[x]/(m(x))$
as algebras.

\subsection{Polynomial functional calculus}
\label{sec:poly-calc}

\begin{lemma}
  \label{prop:twirl1-generates}
  The following hold for all $b, a \in \{0,\dots,n\}$:
  \begin{enumerate}[(i)]
    \item\label{item:lagrange}
      The Lagrange interpolant
      \begin{equation}
        \label{eq:lagrange-def}
        \ell_a(x) := \prod_{\substack{0 \le c \le n \\ c \ne a}}
        \frac{x - x_c}{x_a - x_c}
      \end{equation}
      satisfies $\ell_a(\twirl_1) = \proj_a$.
    \item\label{item:poly-rep}
      There is a unique polynomial $p_b \in \mathbb{C}[x]$ of degree
      at most $n$ such that
      \begin{equation}
        \label{eq:twirl-poly-realization}
        \twirl_b = p_b(\twirl_1),
      \end{equation}
      and this polynomial evaluates as
      \begin{equation}
        \label{eq:pb-values}
        p_b(x_a) = M_{ba} \qquad (a = 0,1,\dots,n).
      \end{equation}
  \end{enumerate}
\end{lemma}

\begin{proof}
  \textit{Closure under polynomial calculus.}
  Since $\twirl_1$ is $G$-equivariant and $\mathcal{A} =
  \mathrm{Hom}_G(\L(V_n), \L(V_n))$ is closed under composition,
  $\twirl_1^k \in \mathcal{A}$ for all $k \ge 0$, and hence
  $p(\twirl_1) \in \mathcal{A}$ for every polynomial $p$.

  For any $X \in \E_c$ and any polynomial $p$,
  \[
    p(\twirl_1)(X)
    = p\bigl(x_c\bigr)\,X, \numberthis
  \]
  since $\twirl_1|_{\E_c} = x_c I_{\E_c}$ and the identity extends to
  polynomials term by term.  Applying this with $p = \ell_a$:
  \[
    \ell_a(\twirl_1)\big|_{\E_c}
    = \ell_a(x_c)\,I_{\E_c} 
    = \delta_{ac}\,I_{\E_c}. \numberthis
  \]
  Thus $\ell_a(\twirl_1)$ acts as the identity on $\E_a$ and annihilates
  every other sector, so $\ell_a(\twirl_1) = \proj_a$.

  Using part~\ref{item:lagrange} and~\eqref{eq:T-diagonal-spectral},
  \[
    \twirl_b
    = \sum_{a=0}^n M_{ba}\,\proj_a
    = \sum_{a=0}^n M_{ba}\,\ell_a(\twirl_1)
    = p_b(\twirl_1), \numberthis
  \]
  where $p_b(x) := \sum_{a=0}^n M_{ba}\,\ell_a(x)$ has degree at most $n$.

This also proves uniqueness.  Indeed, if \(p(\twirl_1)=0\) with
\(\deg p\le n\), then restricting to \(\E_a\) gives
\(p(x_a)=0\) for every \(a=0,\dots,n\).  Since the \(x_a\)'s are distinct,
a polynomial of degree at most \(n\) with these \(n+1\) roots must vanish
identically.

  Restricting $\twirl_b = p_b(\twirl_1)$ to $\E_a$ and
  using~\eqref{eq:twirl-on-sector}:
  \[
    M_{ba}\,I_{\E_a}
    = \twirl_b\big|_{\E_a}
    = p_b(x_a)\,I_{\E_a}, \numberthis
  \]
  which gives $p_b(x_a) = M_{ba}$.
\end{proof}

By the Plancherel identity \eqref{eq:plancherel} and the evaluation formula
\(M_{ba}=p_b(x_a)\), the row polynomials satisfy
\begin{equation}
\label{eq:poly-orthogonality}
\sum_{a=0}^n d_a\,\overline{p_b(x_a)}\,p_c(x_a)
=
d_b\,\delta_{bc}.
\end{equation}
Thus \(\{p_b\}\) is an orthogonal family on the spectral grid
\(\{x_a\}_{a=0}^n\) with weights \(d_a\).

\section{Tridiagonality of Multiplication by \texorpdfstring{$\twirl_1$}{T\_1}}
\label{sec:suq-support}

The goal of this section is to prove that multiplication by the degree-one
twirl is tridiagonal in the twirl basis:
\[
  \twirl_1\twirl_b
  \in
  \mathrm{span}\{\twirl_{b-1},\twirl_b,\twirl_{b+1}\},
  \qquad b=0,1,\dots,n.
\]
The proof has two ingredients.  First, a Pieri-rule calculation shows that
the only diagonal sectors \(\E_r=(r,0,\dots,0,r)\) that can occur in
\(\E_b\otimes\Ad\) have \(r\in\{b-1,b,b+1\}\).  Second, the
\(G\)-equivariance of operator multiplication,
\[
  m:\E_b\otimes\E_c\to\L(V_n),\qquad m(E\otimes F)=EF,
\]
implies that the expansion of \(\twirl_c\twirl_b\) in the twirl basis can
involve only those \(\twirl_r\) for which \(\E_r\) occurs in
\(\E_b\otimes\E_c\).  Applying this with \(c=1\) and
\(\E_1\cong\Ad\) gives the result.

\subsection{Pieri rules and the restricted support lemma}
\label{sec:pieri-support}

Let
\begin{align}
   \F &:= (1,0,\dots,0), \\
   \F^* &:= (0,\dots,0,1), \\
  \Ad &:= (1,0,\dots,0,1) 
\end{align}
denote the defining representation, its dual, and the adjoint representation
of $G = \SU(q)$, respectively. Throughout this subsection, a Dynkin-label tuple
\((a_1,\dots,a_{q-1})\) denotes the irreducible \(\SU(q)\)-representation
with highest weight \(a_1\omega_1+\cdots+a_{q-1}\omega_{q-1}\). We use the Pieri rules, equivalently the special case of the
Littlewood--Richardson rule, for tensoring with the defining representation
or its dual \cite{fultonharris}.
 For an irreducible highest-weight module
$(a_1,\dots,a_{q-1})$ with $a_i \in \mathbb{Z}_{\ge 0}$
\begin{align}
  \label{eq:suq-pieri-fund}
  (a_1,\dots,a_{q-1}) \otimes F
  \;\cong\;&
  (a_1+1,\,a_2,\dots,a_{q-1})
  \nonumber\\
  &\oplus\,
  (a_1-1,\,a_2+1,\,a_3,\dots,a_{q-1})
  \nonumber\\
  &\oplus\cdots\oplus\,
  (a_1,\dots,a_{q-2},\,a_{q-1}-1)
\end{align}
where any summand with a negative entry is omitted.
The dual rule is obtained from~\eqref{eq:suq-pieri-fund} by the involution
of the Dynkin diagram $A_{q-1}$ that reverses node order, mapping
$\omega_k \mapsto \omega_{q-k}$ and in particular $\F = \omega_1$ to
$\F^* = \omega_{q-1}$:
\begin{align}
  \label{eq:suq-pieri-dual}
  (a_1,\dots,a_{q-1}) \otimes \F^*
  \;\cong\;&
  (a_1,\dots,a_{q-2},\,a_{q-1}+1)
  \nonumber\\
  &\oplus\,
  (a_1,\dots,a_{q-3},\,a_{q-2}+1,\,a_{q-1}-1)
  \nonumber\\
  &\oplus\cdots\oplus\,
  (a_1-1,\,a_2,\dots,a_{q-1}),
\end{align}
where again any summand with a negative entry is omitted.

We also use the standard decomposition
\begin{equation}
  \label{eq:FFstar}
  \F \otimes \F^* \;\cong\; \mathbf{1} \oplus \Ad,
\end{equation}
which follows from $\mathbb{C}^q \otimes (\mathbb{C}^q)^* \cong
\mathrm{End}(\mathbb{C}^q) \cong \mathfrak{sl}(q) \oplus \mathbb{C}$,
the traceless part being the adjoint.

\begin{lemma}[Restricted diagonal support]
\label{lem:suq-diagonal-support}
Let $q \ge 2$ and $b, r \ge 0$.  If
$\Hom_{\SU(q)}\!\bigl(\E_r,\,\E_b \otimes \Ad\bigr) \ne 0$,
then $r \in \{b-1, b, b+1\}$, omitting $b-1$ when $b = 0$.
\end{lemma}

\begin{proof}
\textit{Case $q = 2$.}
For $\SU(2)$ we are in a degenerate case where $ \E_b \cong (2b)$ and 
$\Ad \cong (2)$, the claim reduces to the standard
Clebsch--Gordan decomposition
\[
  (2b) \otimes (2) \cong (2b+2) \oplus (2b) \oplus (2b-2), \numberthis
\]
which gives $r \in \{b-1, b, b+1\}$ immediately.

\textit{Case $q \ge 3$.}
From~\eqref{eq:FFstar},
$\E_b \otimes \F \otimes \F^* \cong \E_b \otimes (\mathbf{1} \oplus \Ad)
\cong \E_b \oplus (\E_b \otimes \Ad)$,
so every irreducible constituent of $\E_b \otimes \Ad$ also occurs in
$\E_b \otimes F \otimes F^*$.
It therefore suffices to identify which sectors $\E_r = (r,0,\dots,0,r)$
occur in $\E_b \otimes \F \otimes \F^*$.

Applying~\eqref{eq:suq-pieri-fund} to $\E_b = (b,0,\dots,0,b)$: every
term at an interior position $k$ ($2 \le k \le q{-}2$) would require
decrementing a Dynkin label that is already zero, so it is omitted.
Exactly three terms survive:
\begin{equation}
  \label{eq:Vb-times-F}
  \E_b \otimes F \;\cong\; A_b \oplus B_b \oplus C_b,
\end{equation}
where $A_b := (b{+}1,0,\dots,0,b)$, $B_b := (b{-}1,1,0,\dots,0,b)$,
$C_b := (b,0,\dots,0,b{-}1)$, and $B_b$ is omitted when $b = 0$.
We now apply~\eqref{eq:suq-pieri-dual} to each summand and retain only
constituents $(r_1,0,\dots,0,r_{q-1})$ with $r_1 = r_{q-1}$ and all
interior labels zero, i.e.\ those of the form $\E_r$.

\smallskip
\noindent\emph{From $A_b = (b{+}1, 0,\dots,0, b)$.}
Incrementing the last label gives $\E_{b+1}$; decrementing the first
gives $\E_b$.  All intermediate terms acquire a nonzero interior label.
Contribution: $\{\E_b, \E_{b+1}\}$.

\smallskip
\noindent\emph{From $B_b = (b{-}1, 1, 0,\dots,0, b)$.}
The unique term that clears the interior label increments position $1$
and decrements position $2$:
\[
  \bigl((b{-}1){+}1,\; 1{-}1,\; 0,\dots,0,\; b\bigr)
  = (b, 0,\dots,0, b) = \E_b.
\]
All other terms retain a nonzero label or have a negative entry.
Contribution: $\{\E_b\}$.

\smallskip
\noindent\emph{From $C_b = (b, 0,\dots,0, b{-}1)$.}
Incrementing the last label gives $\E_b$; decrementing the first gives
$\E_{b-1}$.  All intermediate terms acquire a nonzero interior label.
Contribution: $\{\E_{b-1}, \E_b\}$.

\smallskip
Combining all three cases, every $\E_r$ occurring in
$\E_b \otimes \F \otimes \F^*$ satisfies $r \in \{b{-}1, b, b{+}1\}$.
\end{proof}
  
\subsection{Tridiagonality}
\label{sec:tridiagonality-proof}

We now pass from the representation-theoretic support statement to the
twirl-basis product \(\twirl_1\twirl_b\).

\begin{lemma}[Twirl-product support]
\label{lem:twirl-product-support}
Let \(0\le b,c\le n\), and write
\[
  \twirl_c\twirl_b=\sum_{r=0}^n N_{cb}^r\twirl_r.
\]
If $ \Hom_G(\E_r,\E_b\otimes\E_c)=0,$ then \(N_{cb}^r=0\).
\end{lemma}

\begin{proof}
Let \(\{E_i\}_{i=1}^{d_b}\) and \(\{F_j\}_{j=1}^{d_c}\) be orthonormal
bases of \(\E_b\) and \(\E_c\), respectively.  Then
\[
  \twirl_c\twirl_b(X)
  =
  \sum_{i,j}(E_iF_j)^\dagger X(E_iF_j).
\]
The products \(E_iF_j\) are the images of the basis elements
\(E_i\otimes F_j\) under the multiplication map
\[
  m:\E_b\otimes\E_c\longrightarrow \L(V_n),
  \qquad
  m(E\otimes F)=EF.
\]
This map is \(G\)-equivariant for the conjugation action, since
\[
  g\cdot(EF)=(g\cdot E)(g\cdot F).
\]
Hence \(m(\E_b\otimes\E_c)\) can have components only in those irreducible
sectors \(\E_r\subset\L(V_n)\) whose type occurs in \(\E_b\otimes\E_c\).

If
\(\Hom_G(\E_r,\E_b\otimes\E_c)=0\), then equivariance of \(m\) implies
\[
  \proj_r(E_iF_j)=0
  \qquad\text{for all }i,j.
\]
where \(\proj_r\) is the orthogonal projector onto $ \E_r$.
Thus the products \(E_iF_j\) have no \(\E_r\)-component.  Therefore the
twirl built from these products has no \(\twirl_r\)-component in its
twirl-basis expansion, and hence \(N_{cb}^r=0\).
\end{proof}
\begin{lemma}[Tridiagonality of \(\twirl_1\)]
\label{lem:tridiagonality}
For each \(b=0,1,\dots,n\),
\[
  \twirl_1\twirl_b
  \in
  \mathrm{span}\{\twirl_{b-1},\twirl_b,\twirl_{b+1}\},
\]
with the convention \(\twirl_{-1}=\twirl_{n+1}=0\).  Equivalently, if
\[
  \twirl_1\twirl_b=\sum_{r=0}^n N_{1b}^r\twirl_r,
\]
then \(N_{1b}^r=0\) whenever \(r\notin\{b-1,b,b+1\}\).
\end{lemma}

\begin{proof}
By \Cref{lem:twirl-product-support}, \(N_{1b}^r=0\) unless
\[
  \Hom_G(\E_r,\E_b\otimes\E_1)\ne0.
\]
Since \(\E_1\cong\Ad\), \Cref{lem:suq-diagonal-support} implies that this
can occur only for
\[
  r\in\{b-1,b,b+1\},
\]
with the convention that \(b-1\) is omitted when \(b=0\).  This proves the
claim.
\end{proof}

\section{From Tridiagonality to Orthogonal Polynomials}
\label{sec:tridiagonal-ops}

We now translate the tridiagonality of \Cref{lem:tridiagonality} into a
three-term recurrence for the polynomial representatives \(p_b\).  The key
point is that tridiagonality alone gives a recurrence on the finite spectral
grid, while nonvanishing of the forward coefficient implies
\(\deg p_b=b\).  Combined with the Plancherel orthogonality from
\eqref{eq:poly-orthogonality}, this shows that
\(\{p_b\}_{b=0}^n\) is a finite orthogonal polynomial system.
Throughout this section, write
\[
  \twirl_1\twirl_b=\sum_{r=0}^n N_{1b}^r\twirl_r.
\]
By \Cref{lem:tridiagonality}, this reduces to
\[
  \twirl_1\twirl_b
  =
  N_{1b}^{b+1}\twirl_{b+1}
  +
  N_{1b}^{b}\twirl_b
  +
  N_{1b}^{b-1}\twirl_{b-1},
\]
with \(\twirl_{-1}=\twirl_{n+1}=0\).

\subsection{Nonvanishing of the forward coefficient}
\begin{lemma}
  \label{lem:forward-nonzero}
  For every \(b=0,1,\dots,n-1\), one has \(N_{1b}^{b+1}\ne0\).
\end{lemma}

\begin{proof}
Suppose, for contradiction, that \(N_{1b}^{b+1}=0\) for some \(b<n\), and set
\[
  U_b:=\operatorname{span}\{\twirl_0,\twirl_1,\dots,\twirl_b\}\subset\mathcal A.
\]
We claim that \(U_b\) is invariant under left multiplication by \(\twirl_1\).
For \(j<b\), tridiagonality gives
\[
  \twirl_1\twirl_j
  =
  N_{1j}^{j+1}\twirl_{j+1}
  +
  N_{1j}^{j}\twirl_j
  +
  N_{1j}^{j-1}\twirl_{j-1},
\]
and all terms on the right lie in \(U_b\).  For \(j=b\), the assumption
\(N_{1b}^{b+1}=0\) gives
\[
  \twirl_1\twirl_b
  =
  N_{1b}^{b}\twirl_b
  +
  N_{1b}^{b-1}\twirl_{b-1}
  \in U_b.
\]
Thus \(\twirl_1U_b\subseteq U_b\).

Since \(\E_0=\mathbb C I_{V_n}\), we have
\[
  \twirl_0=\frac{1}{\dim V_n}I_{\L(V_n)},
\]
where \(I_{\L(V_n)}\) is the identity superoperator on \(\L(V_n)\).
Hence \(I_{\L(V_n)}\in U_b\).  By \Cref{prop:twirl1-generates},
\(\mathcal A=\mathbb C[\twirl_1]\).  Since \(U_b\) contains the identity
and is invariant under multiplication by \(\twirl_1\), it contains every
polynomial in \(\twirl_1\), hence \(U_b=\mathcal A\).  But
\(\dim U_b=b+1<n+1=\dim\mathcal A\), a contradiction.
\end{proof}

\subsection{The three-term recurrence and the orthogonal polynomial system}

\begin{lemma}
  \label{lem:three-term-and-degrees}
  For each \(b = 0,1,\dots,n-1\), the polynomial identity
  \begin{equation}
    \label{eq:three-term}
    x\,p_b(x)
    =
    N_{1b}^{b+1}\,p_{b+1}(x)
    +
    N_{1b}^{b}\,p_b(x)
    +
    N_{1b}^{b-1}\,p_{b-1}(x)
  \end{equation}
  holds in \(\mathbb C[x]\), with the convention \(p_{-1}=0\).  Moreover,
  \(\deg p_b=b\) for all \(b=0,1,\dots,n\), and
  \(\{p_b\}_{b=0}^n\) is a finite orthogonal polynomial system with respect
  to the discrete measure
  \[
    \mu=\sum_{a=0}^n d_a\,\delta_{x_a}.
  \]
\end{lemma}

\begin{proof}
We prove the recurrence and the degree statement simultaneously by induction
on \(b\).

\emph{Base case.}
Since \(\E_0=\mathbb C I_{V_n}\), an orthonormal basis of \(\E_0\) is
\[
  \left\{\frac{I_{V_n}}{\sqrt{\dim V_n}}\right\}.
\]
Therefore
\[
  \twirl_0(X)
  =
  \left(\frac{I_{V_n}}{\sqrt{\dim V_n}}\right)^\dagger
  X
  \left(\frac{I_{V_n}}{\sqrt{\dim V_n}}\right)
  =
  \frac{1}{\dim V_n}X.
\]
Thus, as a superoperator,
\[
  \twirl_0=\frac{1}{\dim V_n}I_{\L(V_n)}.
\]
Hence \(p_0(x)=1/\dim V_n\), so \(\deg p_0=0\).

\emph{Inductive step.}
Assume \(\deg p_j=j\) for all \(j\le b\), where \(b<n\).  We have
\begin{equation}
  \label{eq:inductive-q}
  N_{1b}^{b+1}\,\twirl_{b+1}
  =
  \twirl_1\twirl_b
  -
  N_{1b}^{b}\,\twirl_b
  -
  N_{1b}^{b-1}\,\twirl_{b-1}
  =
  q(\twirl_1),
\end{equation}
where
\begin{equation}
  \label{eq:q-def}
  q(x)
  :=
  x\,p_b(x)
  -
  N_{1b}^{b}\,p_b(x)
  -
  N_{1b}^{b-1}\,p_{b-1}(x).
\end{equation}
By the inductive hypothesis, \(\deg p_b=b\) and
\(\deg p_{b-1}=b-1\).  Thus the subtracted terms in \eqref{eq:q-def} have
degree at most \(b\), while \(x p_b(x)\) has degree \(b+1\).  Hence $ \deg q=b+1. $

Since \(b+1\le n\), both \(q\) and \(N_{1b}^{b+1}p_{b+1}\) are
degree-\(\le n\) representatives of the same element of \(\mathcal A\).
By the uniqueness of degree-\(\le n\) representatives from
\Cref{prop:twirl1-generates},
\[
  N_{1b}^{b+1}p_{b+1}=q
  \qquad\text{in }\mathbb C[x].
\]
By \Cref{lem:forward-nonzero}, \(N_{1b}^{b+1}\ne0\).  Therefore
\(\deg p_{b+1}=b+1\).  Rearranging the same identity gives
\eqref{eq:three-term}.

By induction, \(\deg p_b=b\) for all \(b=0,\dots,n\).  Finally, the
discrete orthogonality relation \eqref{eq:poly-orthogonality} gives
\[
  \int \overline{p_b(x)}p_c(x)\,d\mu(x)
  =
  \sum_{a=0}^n d_a\,\overline{p_b(x_a)}p_c(x_a)
  =
  d_b\delta_{bc}.
\]
Together with \(\deg p_b=b\), this shows that
\(\{p_b\}_{b=0}^n\) is a finite orthogonal polynomial system for \(\mu\).
\end{proof}

Evaluating the polynomial recurrence \eqref{eq:three-term} on the spectral
grid gives the corresponding row recurrence for the MacWilliams matrix.
Since \(p_b(x_a)=M_{ba}\), for \(0\le b\le n-1\) and \(0\le a\le n\),
\begin{equation}
  \label{eq:matrix-recurrence}
  x_a\,M_{ba}
  =
  N_{1b}^{b+1}\,M_{b+1,a}
  +
  N_{1b}^{b}\,M_{ba}
  +
  N_{1b}^{b-1}\,M_{b-1,a},
\end{equation}
with the convention \(M_{-1,a}=0\).

\section{Identification with Racah Polynomials}
\label{sec:suq-racah}

We now identify the orthogonal polynomial system constructed above with a
Racah polynomial system.  The identification uses only the spectral lattice,
the orthogonality weights, and the normalization at the trivial sector.

Recall that the eigenvalues of \(\twirl_1\) have
the affine form
\begin{equation}
\label{eq:racah-x-y}
x_a=A-B\,y_a,
\qquad
y_a:=a(a+q-1),
\qquad
a=0,1,\dots,n,
\end{equation}
where
\[
A=\frac{q^2-1}{\dim V_n},
\qquad
B=\frac{q(q^2-1)}{\dim V_n\,n(q-1)(n+q)}. \numberthis
\]
Define
\begin{equation}
\label{eq:ptilde-def}
\widetilde p_b(y):=p_b(A-By).
\end{equation}
By Lemma~\ref{lem:three-term-and-degrees} the polynomial \(p_b\) has degree
\(b\). Hence \(\widetilde p_b\) also has degree \(b\). Moreover,
\eqref{eq:poly-orthogonality} becomes
\begin{equation}
\label{eq:ptilde-orthogonality}
\sum_{a=0}^n d_a\,
\overline{\widetilde p_b(y_a)}\,\widetilde p_c(y_a)
=
d_b\,\delta_{bc}.
\end{equation}
Thus \(\{\widetilde p_b\}_{b=0}^n\) is an orthogonal polynomial system for
the positive discrete measure
\begin{equation}
\label{eq:racah-mu-y}
\mu_y:=\sum_{a=0}^n d_a\,\delta_{y_a}.
\end{equation}

We compare this measure with the standard Racah system in the normalization
recorded in Appendix~\ref{app:racah}.  Using the Racah
parameters
\begin{equation}
\label{eq:racah-params}
\alpha=q-2,
\qquad
\beta=0,
\qquad
\gamma=n+q-1,
\qquad
\delta=-n-1
\end{equation}
then
\[
\gamma+\delta+1=q-1,
\qquad
\beta+\delta+1=-n. \numberthis
\]
Therefore the Racah lattice
\[
\lambda(a)=a(a+\gamma+\delta+1) \numberthis
\]
becomes
\begin{equation}
\label{eq:racah-lattice-match}
\lambda(a)=a(a+q-1)=y_a,
\qquad
a=0,1,\dots,n.
\end{equation}
The finite support condition is \(\beta+\delta+1=-n\), so the Racah support
is exactly \(a=0,1,\dots,n\).

\begin{lemma}[Racah weight]
\label{lem:racah-weight}
For the parameters~\eqref{eq:racah-params}, the Racah orthogonality weight
on \(a=0,1,\dots,n\) is proportional to
\[
d_a
=
\frac{2a+q-1}{q-1}
\binom{a+q-2}{q-2}^{2}. \numberthis
\]
Consequently, the Racah orthogonality measure agrees with \(\mu_y\) up to
an overall scalar factor.
\end{lemma}

\begin{proof}
This is the specialization of the standard Racah weight recorded in
Appendix~\ref{app:racah}.  Substituting
\[
\alpha=q-2,\qquad
\beta=0,\qquad
\gamma=n+q-1,\qquad
\delta=-n-1 \numberthis
\]
gives, up to a scalar independent of \(a\),
\[
w(a)
\propto
\frac{2a+q-1}{q-1}
\frac{(q-1)_a^2}{(a!)^2}. \numberthis
\]
Since
\[
(q-1)_a=\frac{(a+q-2)!}{(q-2)!}, \numberthis
\]
we obtain
\[
w(a)
\propto
\frac{2a+q-1}{q-1}
\binom{a+q-2}{q-2}^{2}.\numberthis
\]
This is precisely \(d_a\), by the dimension formula for
\(\E_a\cong(a,0,\dots,0,a)\).
\end{proof}

It remains to fix the normalization.  We use the value at the trivial
sector.

\begin{lemma}[Normalization at the trivial sector]
\label{lem:Mb0}
For every \(b=0,1,\dots,n\),
\begin{equation}
\label{eq:Mb0}
M_{b0}=\frac{d_b}{\dim V_n}.
\end{equation}
\end{lemma}

\begin{proof}
The sector \(\E_0\) is the trivial sector spanned by \(I_{V_n}\).  Since
\(\twirl_b\) is \(G\)-equivariant, it acts on \(\E_0\) by the scalar
\(M_{b0}\):
\[
\twirl_b(I_{V_n})=M_{b0}I_{V_n}. \numberthis
\]
Taking traces gives
\[
M_{b0}\dim V_n=\Tr(\twirl_b(I_{V_n})). \numberthis
\]
Let \(\{E_i\}_{i=1}^{d_b}\) be an orthonormal basis of \(\E_b\).  By the
definition of \(\twirl_b\),
\[
\Tr(\twirl_b(I_{V_n}))
=
\sum_{i=1}^{d_b}\Tr(E_i^\dagger E_i)
=
d_b. \numberthis
\]
Therefore \(M_{b0}=d_b/\dim V_n\).
\end{proof}

Define the normalized row polynomial
\begin{equation}
\label{eq:Qb-def}
Q_b(y):=\frac{\dim V_n}{d_b}\,\widetilde p_b(y)
=
\frac{\dim V_n}{d_b}\,p_b(A-By).
\end{equation}
Then
\[
Q_b(y_a)=\frac{\dim V_n}{d_b}M_{ba}. \numberthis
\]
Since \(y_0=0\), Lemma~\ref{lem:Mb0} gives
\begin{equation}
\label{eq:Qb-y0}
Q_b(y_0)=1.
\end{equation}
Moreover \(Q_b\) has degree \(b\) and is orthogonal to every polynomial of
degree less than \(b\) with respect to \(\mu_y\).

Let \(R_b(y)\) denote the Racah polynomial with parameters
\eqref{eq:racah-params}, normalized by
\[
R_b(y_0)=1. \numberthis
\]
With the convention of Appendix~\ref{app:racah}, its values on the lattice
\(y_a=a(a+q-1)\) are
\begin{equation}
\label{eq:Rb-hypergeom}
R_b(y_a)
=
{}_4F_3\!\left(
\begin{matrix}
-b,\ b+q-1,\ -a,\ a+q-1\\
q-1,\ -n,\ n+q
\end{matrix}
;\,1
\right).
\end{equation}

\begin{theorem}[Racah identification]
\label{thm:racah-identification}
For each \(b=0,1,\dots,n\),
\begin{equation}
\label{eq:Qb-Rb}
Q_b(y)=R_b(y).
\end{equation}
Equivalently,
\begin{equation}
\label{eq:p-racah-identification}
p_b(A-By)
=
\frac{d_b}{\dim V_n}\,R_b(y).
\end{equation}
\end{theorem}

\begin{proof}
The polynomial \(Q_b\) has degree \(b\), satisfies \(Q_b(y_0)=1\), and is
orthogonal to all polynomials of degree less than \(b\) with respect to
\(\mu_y\).  By Lemma~\ref{lem:racah-weight}, the measure \(\mu_y\) is a
nonzero scalar multiple of the Racah orthogonality measure for the parameters
\eqref{eq:racah-params}.  Multiplying a measure by a nonzero scalar does not
change its orthogonal polynomials.  Hence \(Q_b\) and \(R_b\) are
degree-\(b\) orthogonal polynomials for the same finite positive measure, and
both satisfy the normalization \(Q_b(y_0)=R_b(y_0)=1\).  By uniqueness of the
degree-\(b\) orthogonal polynomial with this normalization, recorded in
Appendix~\ref{app:racah}, we have \(Q_b=R_b\).
\end{proof}

\begin{corollary}[Explicit MacWilliams matrix]
\label{cor:explicit-M}
For \(a,b=0,1,\dots,n\),
\begin{equation}
\label{eq:Mba-explicit}
M_{ba}
=
\frac{d_b}{\dim V_n}
\;
{}_4F_3\!\left(
\begin{matrix}
-b,\ b+q-1,\ -a,\ a+q-1\\
q-1,\ -n,\ n+q
\end{matrix}
;\,1
\right)
\end{equation}
Here
\[
d_b
=
\frac{2b+q-1}{q-1}
\binom{b+q-2}{q-2}^{2},
\;
\dim V_n=\binom{n+q-1}{n} \numberthis
\]
\end{corollary}

\begin{proof}
By \eqref{eq:Qb-def},
\[
M_{ba}=\frac{d_b}{\dim V_n}Q_b(y_a). \numberthis
\]
Theorem~\ref{thm:racah-identification} gives \(Q_b=R_b\), and substituting
the Racah value~\eqref{eq:Rb-hypergeom} gives~\eqref{eq:Mba-explicit}.
\end{proof}

Note that for \(b=0\),
\[
M_{0a}=\frac{1}{\dim V_n}. \numberthis
\]
And for \(b=1\), one has \(d_1=q^2-1\), and
\eqref{eq:Mba-explicit} gives
\[
M_{1a}
=
\frac{q^2-1}{\dim V_n}
{}_4F_3\!\left(
\begin{matrix}
-1,\ q,\ -a,\ a+q-1\\
q-1,\ -n,\ n+q
\end{matrix}
;\,1
\right). \numberthis
\]
Since the first numerator parameter is \(-1\), the series terminates after
the linear term:
\[
{}_4F_3\!\left(
\begin{matrix}
-1,\ q,\ -a,\ a+q-1\\
q-1,\ -n,\ n+q
\end{matrix}
;\,1
\right)
=
1-\frac{q\,a(a+q-1)}{n(q-1)(n+q)}. \numberthis
\]
Thus
\[
M_{1a}
=
\frac{q^2-1}{\dim V_n}
\left(
1-\frac{q\,a(a+q-1)}{n(q-1)(n+q)}
\right), \numberthis
\]
which agrees with the direct Casimir computation of the degree-one twirl in
Section~\ref{sec:x_a}. 

Also, for \(q=2\), one has
\[
\dim V_n=n+1,\qquad d_a=2a+1,\qquad y_a=a(a+1). \numberthis
\]
Thus
\begin{equation}
\label{eq:su2-M-clean}
M_{ba}
=
\frac{2b+1}{n+1}
\,{}_4F_3\!\left(
\begin{matrix}
-b,\ b+1,\ -a,\ a+1\\
1,\ -n,\ n+2
\end{matrix}
;\,1
\right).
\end{equation}
This is the Racah system underlying the corresponding Wigner \(6j\)-symbol
formula agreeing with direct computation given in \cite{usIntrinsicMacWilliams}.

\section{Properties of The MacWilliams Transform}
\label{sec:structure}

\subsection{Orthogonality, detailed balance, and involutivity}
\label{sec:orth-inverse}

We collect the structural identities satisfied by the intrinsic MacWilliams
matrix. Let $ D $ denote the diagonal matrix with diagonal entries $ d_0,\dots,d_n $.

\begin{proposition}[Orthogonality]
\label{prop:orthogonality}
The intrinsic MacWilliams matrix satisfies
\begin{equation}
\label{eq:orthogonality-clean}
M D M^T = D.
\end{equation}
Equivalently, for all \(b,c=0,1,\dots,n\),
\begin{equation}
\label{eq:orthogonality-entry}
\sum_{a=0}^n d_a M_{ba}M_{ca}=d_b\delta_{bc}.
\end{equation}
\end{proposition}

\begin{proof}
This is the Plancherel identity~\eqref{eq:plancherel} established in
Section~\ref{sec:setup}, together with the realness of \(M\), which follows
from the explicit formula in Corollary~\ref{cor:explicit-M}.
\end{proof}

\begin{proposition}[Inverse transform]
\label{prop:inverse-transform}
The inverse of \(M\) is
\begin{equation}
\label{eq:M-inverse-clean}
M^{-1}=D M^T D^{-1}.
\end{equation}
Equivalently,
\begin{equation}
\label{eq:M-inverse-entry-clean}
(M^{-1})_{ab}
=
\frac{d_a}{d_b}M_{ba}.
\end{equation}
\end{proposition}

\begin{proof}
From \(MDM^T=D\), right-multiplication by \(D^{-1}\) gives
\[
M D M^T D^{-1}=I. \numberthis
\]
Thus
\[
M^{-1}=D M^T D^{-1}. \numberthis
\]
Taking the \((a,b)\)-entry gives
\[
(M^{-1})_{ab}
=
D_{aa}(M^T)_{ab}(D^{-1})_{bb}
=
d_a M_{ba}\frac{1}{d_b}
=
\frac{d_a}{d_b}M_{ba}.
\]
\end{proof}

\begin{proposition}[Detailed balance]
\label{prop:detailed-balance}
For all \(a,b=0,1,\dots,n\),
\begin{equation}
\label{eq:detailed-balance-clean}
d_aM_{ba}=d_bM_{ab}.
\end{equation}
Equivalently,
\begin{equation}
\label{eq:DMt-MD}
D M^T = M D.
\end{equation}
\end{proposition}

\begin{proof}
By the explicit formula~\eqref{eq:Mba-explicit},
\[
M_{ba}
=
\frac{d_b}{\dim V_n}
\;
{}_4F_3\!\left(
\begin{matrix}
-b,\ b+q-1,\ -a,\ a+q-1\\
q-1,\ -n,\ n+q
\end{matrix}
;\,1
\right). \numberthis
\]
Therefore
\[
d_aM_{ba}
=
\frac{d_ad_b}{\dim V_n}
\;
{}_4F_3\!\left(
\begin{matrix}
-b,\ b+q-1,\ -a,\ a+q-1\\
q-1,\ -n,\ n+q
\end{matrix}
;\,1
\right). \numberthis
\]
The hypergeometric series is symmetric under interchanging the two numerator
pairs \cite{AskeyWilson1979,KoekoekSwarttouw1998}
\[
(-b,b+q-1)
\qquad\text{and}\qquad
(-a,a+q-1). \numberthis
\]
Hence the last expression is symmetric in \(a\) and \(b\), proving
\[
d_aM_{ba}=d_bM_{ab}. \numberthis
\]
\end{proof}

\begin{corollary}[Involutivity]
\label{cor:M-involutive}
The intrinsic MacWilliams matrix is an involution:
\begin{equation}
\label{eq:M-involutive}
M^2=I.
\end{equation}
Equivalently,
\[
M^{-1}=M. \numberthis
\]
\end{corollary}

\begin{proof}
By Proposition~\ref{prop:inverse-transform},
\[
M^{-1}=DM^TD^{-1}. \numberthis
\]
By detailed balance, \(DM^T=MD\). Hence
\[
M^{-1}=DM^TD^{-1}=MDD^{-1}=M. \numberthis
\]
Therefore \(M^2=I\).
\end{proof}

\begin{theorem}[Racah structure of the intrinsic MacWilliams transform]
\label{thm:main-clean}
Let
\[
V_n=\Sym^n(\CC^q),
\qquad
\dim V_n=\binom{n+q-1}{n}, \numberthis
\]
and let
\[
\L(V_n)=\bigoplus_{a=0}^n \E_a,
\qquad
\E_a\cong(a,0,\dots,0,a), \numberthis
\]
be the multiplicity-free decomposition under conjugation by \(\SU(q)\).
Set
\[
d_a=\dim \E_a
=
\frac{2a+q-1}{q-1}
\binom{a+q-2}{q-2}^{2}. \numberthis
\]
Then the intrinsic MacWilliams matrix \(M=(M_{ba})_{0\le a,b\le n}\) is
given by
\begin{equation}
\label{eq:main-Mba}
M_{ba}
=
\frac{d_b}{\dim V_n}
\;
{}_4F_3\!\left(
\begin{matrix}
-b,\ b+q-1,\ -a,\ a+q-1\\
q-1,\ -n,\ n+q
\end{matrix}
;\,1
\right).
\end{equation}
The degree-one row is the affine quadratic spectral grid
\[
M_{1a}
=
\frac{q^2-1}{\dim V_n}
\left(
1-
\frac{q\,a(a+q-1)}{n(q-1)(n+q)}
\right). \numberthis
\]
Moreover, 
\[
MDM^T=D,
\qquad
M^2=I. \numberthis
\]
\end{theorem}

\begin{proof}
The decomposition and sector dimensions are established in
Section~\ref{sec:setup}. The spectral computation of the degree-one row is
Lemma~\ref{lem:twirl1-eigenvalues}.  Orthogonality, detailed balance, and involutivity are
Propositions~\ref{prop:orthogonality}, \ref{prop:detailed-balance}, and
Corollary~\ref{cor:M-involutive}.
\end{proof}

\section*{Acknowledgments}

Thank you to Eric Kubischta for helpful discussions and for reviewing an early version of this manuscript.

\bibliographystyle{IEEEtran}
\bibliography{biblio}

\appendices

\makeatletter
\renewcommand{\theequation}{A\arabic{equation}}
\renewcommand{\thetable}{A\arabic{table}}
\renewcommand{\thefigure}{A\arabic{figure}}
\renewcommand{\thelemma}{A\arabic{lemma}}
\renewcommand{\thetheorem}{A\arabic{theorem}}
\setcounter{table}{0}
\setcounter{figure}{0}
\setcounter{lemma}{0}
\setcounter{theorem}{0}
\setcounter{equation}{0}

\section{Casimir Eigenvalues and Normalization}
\label{app:casimir}

This appendix records the Casimir eigenvalues used in
Section~\ref{sec:x_a}.  We use the Hermitian-generator convention for
\(\mathfrak{su}(q)\): thus \(\mathfrak{su}(q)\) denotes the real vector
space of traceless Hermitian \(q\times q\) matrices, with Lie bracket
\(i[A,B]\).  The generators are normalized by
\begin{equation}
  \label{eq:app-generator-normalization}
  \Tr(T_\mu T_\nu)=\delta_{\mu\nu}.
\end{equation}

Let \(W_\lambda\) be the irreducible \(\SU(q)\)-module of highest weight
\(\lambda\).  With the normalization \eqref{eq:app-generator-normalization},
the quadratic Casimir
\[
  \sum_{\mu=1}^{q^2-1}d\rho_\lambda(T_\mu)^2
\]
acts on \(W_\lambda\) as the scalar
\begin{equation}
  \label{eq:app-casimir-formula}
  c(\lambda)=(\lambda,\lambda+2\rho)_{\rm rt},
\end{equation}
where \((\cdot,\cdot)_{\rm rt}\) is the standard type \(A_{q-1}\) inner
product normalized so that all roots have squared length \(2\), and
\(\rho\) is the Weyl vector.

We use the standard type \(A_{q-1}\) identities
\begin{equation}
  \label{eq:app-fundamental-inner-product}
  (\omega_i,\omega_j)_{\rm rt}
  =
  \min(i,j)-\frac{ij}{q},
  \qquad
  \rho=\omega_1+\cdots+\omega_{q-1}.
\end{equation}
In particular,
\[
  (\omega_1,\omega_1)_{\rm rt}=\frac{q-1}{q},
  \qquad
  (\omega_1,\rho)_{\rm rt}=\frac{q-1}{2}.
\]

For the physical representation
\[
  V_n=\Sym^n(\mathbb C^q),
\]
the highest weight is \(\lambda_V=n\omega_1\).  Hence
\begin{align}
  c_V
  &=
  (n\omega_1,n\omega_1+2\rho)_{\rm rt} \nonumber\\
  &=
  n^2(\omega_1,\omega_1)_{\rm rt}
  +
  2n(\omega_1,\rho)_{\rm rt} \nonumber\\
  &=
  n^2\frac{q-1}{q}+n(q-1)
  =
  \frac{n(q-1)(n+q)}{q}.
\end{align}
Thus
\begin{equation}
  \label{eq:app-cV}
  c_V=\frac{n(q-1)(n+q)}{q}.
\end{equation}

For the conjugation sector
\[
  \E_a\cong(a,0,\dots,0,a),
\]
the highest weight is
\[
  \lambda_a=a\omega_1+a\omega_{q-1}.
\]
Using
\[
  \omega_1+\omega_{q-1}=\varepsilon_1-\varepsilon_q
\]
in the traceless weight space, we may write
\[
  \lambda_a=a(\varepsilon_1-\varepsilon_q).
\]
Since
\[
  (\varepsilon_i,\varepsilon_j)_{\rm rt}
  =
  \delta_{ij}-\frac1q,
\]
we have
\[
  (\lambda_a,\lambda_a)_{\rm rt}=2a^2.
\]
Also, for type \(A_{q-1}\),
\[
  2\rho=\sum_{1\le i<j\le q}(\varepsilon_i-\varepsilon_j),
\]
so
\[
  (\lambda_a,2\rho)_{\rm rt}=2a(q-1).
\]
Therefore
\begin{align}
  c_a
  &:=
  c(\lambda_a)
  =
  (\lambda_a,\lambda_a+2\rho)_{\rm rt} \nonumber\\
  &=
  2a^2+2a(q-1)
  =
  2a(a+q-1).
\end{align}
Thus
\begin{equation}
  \label{eq:app-ca}
  c_a=2a(a+q-1).
\end{equation}

Equivalently, the conjugation Casimir
\[
  \mathscr C(X)=\sum_{\mu=1}^{q^2-1}[J_\mu,[J_\mu,X]]
\]
acts on \(\E_a\) as
\begin{equation}
  \label{eq:app-conj-casimir-on-Ea}
  \mathscr C\big|_{\E_a}=2a(a+q-1)I_{\E_a}.
\end{equation}

\section{Racah Polynomial Convention}
\label{app:racah}

This appendix records the Racah polynomial convention used in
Section~\ref{sec:suq-racah}. We use the standard Racah normalization from Section 9.2 of
\cite{KoekoekLeskySwarttouw2010}.  We follow the normalization in which the
Racah polynomial is equal to \(1\) at the first lattice point.

\subsection{Definition}

The Racah polynomial of degree \(b\) is a polynomial in the quadratic
lattice variable
\[
  \lambda(a)=a(a+\gamma+\delta+1)
\]
defined by
\begin{equation}
\label{eq:app-racah-def}
R_b(\lambda(a))
=
{}_4F_3\!\left(
\begin{matrix}
-b,\ b+\alpha+\beta+1,\ -a,\ a+\gamma+\delta+1\\
\alpha+1,\ \beta+\delta+1,\ \gamma+1
\end{matrix}
;1
\right).
\end{equation}
The series terminates because of the numerator parameter \(-b\), and the
normalization is
\[
  R_b(\lambda(0))=1.
\]

In this paper we use
\begin{equation}
\label{eq:app-racah-params}
\alpha=q-2,\qquad
\beta=0,\qquad
\gamma=n+q-1,\qquad
\delta=-n-1.
\end{equation}
Then
\[
  \gamma+\delta+1=q-1,
  \qquad
  \beta+\delta+1=-n,
\]
so the finite support is \(a=0,1,\dots,n\), and the lattice becomes
\begin{equation}
\label{eq:app-racah-lattice}
  \lambda(a)=a(a+q-1).
\end{equation}
With these parameters,
\begin{equation}
\label{eq:app-racah-4f3}
R_b(\lambda(a))
=
{}_4F_3\!\left(
\begin{matrix}
-b,\ b+q-1,\ -a,\ a+q-1\\
q-1,\ -n,\ n+q
\end{matrix}
;1
\right).
\end{equation}

\subsection{Orthogonality weight}

For the Racah convention in \eqref{eq:app-racah-def}, the standard finite
orthogonality weight is proportional to
\[
w(a)
\propto
\frac{2a+\gamma+\delta+1}{\gamma+\delta+1}
\frac{(\gamma+\delta+1)_a(\alpha+1)_a(\beta+\delta+1)_a(\gamma+1)_a}
{a!\,(\gamma+\delta-\alpha+1)_a(\gamma-\beta+1)_a(\delta+1)_a}.
\]
Here \((z)_a\) denotes the Pochhammer symbol.  Substituting
\eqref{eq:app-racah-params} gives, up to a scalar independent of \(a\),
\[
w(a)
\propto
\frac{2a+q-1}{q-1}
\frac{(q-1)_a^2(-n)_a(n+q)_a}
{a!\,(1)_a(n+q)_a(-n)_a}
=
\frac{2a+q-1}{q-1}
\frac{(q-1)_a^2}{(a!)^2}.
\]
Since
\[
  \frac{(q-1)_a}{a!}
  =
  \binom{a+q-2}{q-2},
\]
we obtain
\begin{equation}
\label{eq:app-racah-weight}
w(a)
\propto
\frac{2a+q-1}{q-1}
\binom{a+q-2}{q-2}^{2}.
\end{equation}
Thus the Racah weight is proportional to the sector dimension
\begin{equation}
\label{eq:app-racah-weight-da}
d_a
=
\frac{2a+q-1}{q-1}
\binom{a+q-2}{q-2}^{2}.
\end{equation}

\subsection{Uniqueness}

We use the following elementary uniqueness fact.  Let
\[
  \mu=\sum_{a=0}^n w_a\,\delta_{y_a}
\]
be a positive measure on \(n+1\) distinct points, with \(w_a>0\).  For each
\(0\le b\le n\), there is at most one polynomial \(Q_b\) of degree \(b\)
which is orthogonal to all polynomials of degree \(<b\) and satisfies a
fixed nonzero normalization such as
\[
  Q_b(y_0)=1.
\]
Indeed, if two such polynomials existed, their difference would have degree
at most \(b-1\), would be orthogonal to all polynomials of degree \(<b\),
and hence would be orthogonal to itself; positivity of \(\mu\) forces the
difference to vanish.

Therefore, once the lattice and weight are identified with
\eqref{eq:app-racah-lattice} and \eqref{eq:app-racah-weight}, any
degree-\(b\) orthogonal polynomial normalized by \(Q_b(y_0)=1\) agrees with
the Racah polynomial \eqref{eq:app-racah-4f3}.

\end{document}